\documentclass[noacm, sigconf, natbib=False]{acmart}

\usepackage[noend]{algorithm2e}
\LinesNumbered
\RestyleAlgo{ruled}
\usepackage{amsmath}
 
\usepackage{array}
 
\usepackage{arydshln}
\usepackage{enumitem}
\usepackage{tabularx}
\usepackage{pgf,tikz}
\usetikzlibrary{positioning}
\usepackage{tikz-network}
\usepackage{caption}
\usepackage{subcaption}

\usepackage[utf8]{inputenc}
\usepackage[]{biblatex}
\AtBeginDocument{%
  \providecommand\BibTeX{{%
    \normalfont B\kern-0.5em{\scshape i\kern-0.25em b}\kern-0.8em\TeX}}}

\setcopyright{none}

\newcommand{\ie}{\textit{i.e.,~}}
\newcommand{\eg}{\textit{e.g.,~}}
\renewcommand\footnotetextcopyrightpermission[1]{}

\begin{document}

\title[]{Friends with Costs and Benefits: Community Formation with Myopic, Boundedly-Rational Actors}

\author{Naina Balepur}
\authornote{Authors contributed equally to this research.}
\email{nainab2@illinois.edu}
\orcid{0009-0005-1155-3705}
\affiliation{%
  \institution{University of Illinois Urbana-Champaign}
  \city{Urbana}
  \state{IL}
  \country{USA}
}

\author{Andy Lee}
\email{andy2@illinois.edu}
\orcid{0000-0002-9014-2496}
\authornotemark[1]
\affiliation{%
  \institution{University of Illinois Urbana-Champaign}
  \city{Urbana}
  \state{IL}
  \country{USA}
}

\author{Hari Sundaram}
\email{hs1@illinois.edu}
\orcid{0000-0003-3315-6055}
\affiliation{%
  \institution{University of Illinois Urbana-Champaign}
  \city{Urbana}
  \state{IL}
  \country{USA}
}

\begin{abstract}
In this paper we address how complex social communities emerge from local decisions by individuals with limited attention and knowledge. This problem is critical; if we understand community formation mechanisms, it may be possible to intervene to improve social welfare. We propose an interpretable, novel model for attributed community formation driven by resource-bounded individuals' strategic, selfish behavior. 
In our stylized model, attributed individuals act strategically in two dimensions: attribute and network structure.
Agents are endowed with limited attention, and communication costs limit the number of active connections. In each time step, each agent proposes a new friendship.
Agents then accept proposals, decline proposals, or remove friends, consistent with their strategy to maximize payoff. 
We identify criteria (number of stable triads) for convergence to some community structure and prove that our community formation model converges to a stable network. Ablations justify the ecological validity of our model and show that each aspect of the model is essential. Our empirical results on a physical world microfinance community demonstrate excellent model fits compared to baseline models. 

\end{abstract}

\maketitle

\pagestyle{plain}

\section{Introduction}

How do complex social communities emerge from behaviors of attributed, selfish, resource-bounded individuals (\ie with limited attention and who lack global network knowledge)? 
Developing an ecologically valid community formation mechanism is a challenging problem. If we can do so, we can improve social welfare by intervening via network-specific policy recommendations. For example, improving social welfare is salient in communities where misinformation spreads~\cite{Bak-Coleman2021}. While we understand mechanisms by which global network structure~\cite{Barabasi1999} and decentralized search~\cite{Kleinberg2000a} emerge from individual actions, we lack similar community formation mechanisms. Community \textit{detection} has been extensively studied~\cite{Fortunato2010,Leskovec2009,Lin2011a,Newman2004,Palla2005,Yang2013a}, but there has been limited work explaining how communities form through individual action. 

While non-strategic models~\cite{Erdoes1960,Albert2002,Watts1998,Holland1983} may generate networks with similar structural properties to social networks, they are not explanatory models. Individual choice plays a crucial role in community formation. Individuals have preferences over the structure of these communities; some may prefer dense groups with high trust, others may prefer to connect to a variety of people, maximizing social capital~\cite{Campbell2015,Granovetter1973,Burt1999,Burt2004}. People may also prefer to connect to those who are similar or different from themselves.

We propose a model for \textit{attributed} community formation, where communities emerge through individuals' strategic, selfish behavior. In our stylized model, attributed individuals strategize in two dimensions: who they want to be friends with (homophily vs. heterophily) and what local network structures they desire (highly embedded vs. social capital). We know from Schelling's seminal work~\cite{Schelling1971,Schelling1978} that an individual's preference for homophilous interaction with neighbors can lead to highly segregated communities. Conversely, friendship diversity can provide structural advantages~\cite{Granovetter1973}.
We also know that some individuals prefer highly embedded communities that foster trust~\cite{Granovetter1985}, yet others seek social capital for the structural advantages~\cite{Burt1999}. 
Inspired by the work of Dunbar ~\cite{Dunbar2009}, individuals in our model can support only a limited number of active communications. By triadic closure~\cite{Kossinets2009}, individuals largely rely on their friends to introduce them to potential connections with some connections resulting from random encounters.

In this work, attributed community formation proceeds as follows. First, each agent examines their neighbors' connections as candidates for friendship. One additional node is revealed at random as a candidate, which mimics friendships arising from serendipitous encounters. After every agent proposes to the utility-maximizing candidate, all agents examine incoming proposals for friendship. They then \textit{myopically} choose to do one of three things: 1) accept a proposal, 2) remove an existing friendship, 3) do nothing --- whichever maximizes their payoff. Since all agents have a constant endowment of resources, and communication with friends is costly, dropping a low utility connection may prove strategic. 
In this work, we consider a population-level mixing of strategies; thus, we can view this model as an evolutionary game.  
We then identify criteria for convergence to a community structure. 
Our empirical results on a large microfinance dataset~\cite{banerjee2013village} demonstrate excellent fits to the model. Our contributions are as follows:

\begin{description}[leftmargin=0.25cm]
    \item[Evolutionarily Stable Communities:] We propose a model for community formation that is evolutionarily stable. We define a notion of stability over triads; in contrast prior work considers pairwise stability~\cite{mele2022structural}. We show that for any finite network the number of stable triads is non-decreasing, and so our model must converge to a non-unique equilibrium state.
    \item[Identification of Essential Agent Properties:] In this work, we identify and test necessary behavioral properties of agents. Our agents are resource-constrained, restricted to local network knowledge only, and have preferences over both network structure and neighbor attributes. Prior work does not consider these properties in conjunction. We demonstrate through ablation tests that these three ecologically valid behavioral assumptions are necessary to explain the formation of complex community structures. We elaborate on these model properties in Section \ref{properties}.
    \item[Interpretability:] Our model parameters characterize the nature of the agents in the network in a way that is easy to understand. Prior work on community formation~\cite{christakis2020empirical,mele2022structural} uses decision functions whose parameters are difficult to interpret and rely upon vast amounts of data. We identify mixing proportions over strategies which help us explain and predict the network outcome.
\end{description}

\section{Related Work}

There are many social network formation models that focus on producing graphs with some subset of observed properties of real networks.
Within this broader class of network formation models are models which focus on community formation.

\textit{Non-strategic Network Formation:}
Classic non-strategic network formation models include the Erdős–Rényi \cite{Erdoes1960}, Barabási–Albert \cite{Albert2002}, Watts-Strogatz \cite{Watts1998}, and the Stochastic Block Model \cite{Holland1983}.
The model in \cite{sallaberry2013model} merges cliques to form a network with strong community structure.
While these models may generate networks which have similar structural properties to social networks, they are not plausibly explanatory models due to the lack of individual choice.

\textit{Strategic Network Formation:}
Network formation games \cite{Fabrikant2003, Anshelevich2006, Anshelevich2008, Jackson2005} have been extensively studied and efficiency bounds are often known.
Utility functions in such games are usually interpretable and the behavior of agents understandable.
Such network games include local formation games \cite{Fabrikant2003}, geodisic utility games \cite{Jackson2005}, and information passing games \cite{Jackson2005}.
In the directed setting there is work in discrete choice models of social network formation \cite{Overgoor2019}.
However, agents in these models often have complete network information and no resource bounds. Agents also rarely consider both structure and attributes in their edge formation decisions.

\textit{Best Response Strategies:} Agents taking best response strategies make optimal selfish decisions while making assumptions about others' strategies. Much of the work on best response for network formation \cite{friedrich2017efficient, derks2008local, goeree2009search} builds off the work of Bala and Goyal \cite{bala2000noncooperative}. They present network formation as a non-cooperative game where agents' actions lead to equilibrium social networks with simple structures. We also present a non-cooperative game for network formation, but our agents do not make assumptions about others' strategies, and only observe the current state of the network.

\textit{Bounded Rationality:}
Bounded rationality, first detailed by Simon \cite{Simon1955}, has been studied in a variety of contexts.
It is known that in general agents violate rationality axioms \cite{Nielsen2022} and do not maximize expected utility \cite{Selten1990}.
Based on this prior work we consider agents who do not perfectly maximize expected utility.
Our baseline models (detailed in \cite{christakis2020empirical} and \cite{mele2022structural}) also attempt to incorporate bounded rationality. In both, agents consider the formation of only one edge in each iteration; so while knowledge is indeed limited, the choice of this edge is not strategic.

\textit{Triangle Closing:}
Triadic closure is a commonly observed phenomena in social networks \cite{Granovetter1973}, and prior network formation models have captured this property.
The clique merging model of \cite{sallaberry2013model} produces networks with large numbers of triangles.
Random walk models such as the process described in \cite{Kumpula2009} also generate networks with high triangle count, however this model is not attributed.

\textit{Attribute Assortativity:}
Many network formation models capture homophilic preferences, and this area has been surveyed \cite{Rivera2010}.
Some \cite{Graham2017, Boucher2015} follow empirical methods for estimating the homophilic preferences of agents in network formation.
Others \cite{Currarini2016} follow an approach more similar to our own, focusing on thresholds of assortativity and the impact on the final network. However, we seek a process that also involves preferences over structure. In real world networks, some people may prefer tightly knit groups with high trust, whereas others may prefer to connect to a variety of people, maximizing social capital~\cite{Campbell2015,Granovetter1973,Burt1999,Burt2004}.

\section{Problem Statement}

We model community formation as an evolutionary game amongst strategic, boundedly-rational agents. Consider a world with agent set $V$; each agent $v$ is endowed with immutable attributes $\tau_v$, and a bounded resource. In this game, each agent seeks to form friendships with other agents, which incur a communication cost. Thus, agents are constrained to only maintain at most $\kappa$ connections. At any time $t$, there exists an undirected network $G_t = (V,E_t)$. Due to the resource constraint, each agent must strategically choose friendships to maximize their payoff. 

Agents have latent preferences in $d$ dimensions.
We model this by a $d$ dimensional strategy space $S$. As in ecological games, we assume that for each of the $d$ dimensions, agents pick a pure strategy such as homophily or heterophily. We assume a population-level mixing of strategies $\pi$ over the $d$ dimensions. 
Thus, each agent $v$ has a strategy $s_v \in S$ with each dimension sampled according to $\pi$.

Each agent $v$ chooses actions to maximize their payoff. The payoff of each action is determined by $v$'s strategy $s_v$ and the state of the network $G_t$.  After each agent simultaneously modifies their network by taking an action during time $t$, the resultant network is $G_{t+1}$.
We ask: will a local edge-formation process with strategic, resource-bounded agents converge to a network with stable community properties 
at some time $t^*$ ?

\section{Essential Agent Properties} \label{properties}

We determine that the following agent properties are essential to future work in ecologically-valid, strategic community formation:

\begin{description}[style=unboxed,leftmargin= .25cm]
    \item[Resource-Constrained Agents:] Maintaining active connections is costly, and
    humans are cognitively constrained \cite{Dunbar2009}.
    We introduce a limit on agents' social interactions by way of a degree constraint $\kappa$ --- agents can only maintain $\kappa$-many edges at a time. 
    We implement this by considering each agent to have a social interaction budget of 1, where each edge has cost $c = 1/\kappa$. Thus, we require that the following inequality must be true for all agents $v \in V$: $c*\delta(v)\leq1$, where $\delta(v)$ is the degree of agent $v$. Agents assign no inherent utility to any ``leftover'' budget.
    
    \item[Local Network Knowledge:] In social networks, agents do not know global network properties (\eg the total number of agents, the distance between themselves and distant agents), so we ensure that they only have knowledge of local structure.
    An agent may propose to form an edge with another agent if there is a path of length two between them, or if they were introduced by chance.
    We adopt this restriction because it represents triadic closure \cite{Kossinets2009}, an agent's connection introducing them to their neighbors.

    \item[Core Strategy Space:] People consider many factors when choosing friends.  However, prior work suggests that two factors: homophily (befriending similar others)~\cite{Centola2011,Centola2013a,McPherson2001,Kossinets2009} and local structure (social capital, embeddedness)~\cite{Burt1999,Burt2004,Granovetter1985} are \textit{both} essential strategic considerations. Thus, when forming friendships, our agents consider both the attribute of the friend, and the local network structure induced by the friendship. 

\end{description}

We show through ablations that the properties detailed above are essential to our model. The expressive and diverse communities our model produces are not possible without each of these elements.

\section{Community Formation Model} \label{the_model}

In this section we give a detailed model description.
We begin with the strategy space and the strategies' utility functions.
We give a high level description of the network formation process as well as a pseudo-code description.
Finally we give the stopping criteria for our model and illustrate its validity.
We introduce considerable notation in this section which is summarized in Table ~\ref{tab:notation}.

\begin{table}[H]
    \caption{Summary of Notation}
    \label{tab:notation}
    \begin{tabularx}{\linewidth}{@{}r X @{}}
    \textsc{Notation} & \textsc{Description}\\ 
    \midrule
        $S$ & $d$ dimensional strategy space \\
        $\pi$ & Population-level mixing distribution over $S$ \\
     $s_v \sim \pi$ & Agent $v$'s private strategy drawn using $\pi$\\
     $\tau_v$ & Immutable attribute vector of agent $v$ \\
     $\delta(v)$ & Degree of agent $v$ \\
     $\kappa$ & Degree constraint for all agents \\
     $N(v)$ & Set of nodes in the neighborhood of $v$ \\
     $N^\tau(v)$ & Set of nodes in $N(v)$ with the same type as $v$ \\
     $U_v$ & Utility agent $v$ derives from $N(v)$ \\
     $\Delta_{v}$ & Number of triangles $v$ is part of \\
     $I_v$ & Number independent nodes in $N(v)$ \\
     $d_{G}(u,v)$ & Distance between nodes $u$ and $v$ in graph $G$ \\
     $P^R_v$ & Set of nodes who have proposed to $v$ \\
\end{tabularx}
\end{table}

We consider population-level strategy mixing. 
We assume all agents are present at $t=0$ (no addition or deletion of nodes).
For $|V|$-many agents, the network is parameterized by $\kappa, \pi, \Omega$. We define $\kappa$ as the maximum degree of any agent. 

In this paper we consider two dimensions: attributes and local structure; and exactly two pure strategies in each dimension: homophily (befriend similar), heterophily (befriend dissimilar) for attributes and social capital (maximize the breadth of friends), embeddedness (maximize the number of mutual friends) for local structure. Since our strategy space has two dimensions ($d=2$), the population mixing distribution has two parameters,  $\pi = (\alpha, \beta)$. 

We define $\alpha$ as the proportion of agents adopting homophily, $1-\alpha$ as heterophily, $\beta$ as the proportion of agents adopting social capital, $1-\beta$ as embeddedness, and $\Omega$ as the distribution over immutable attributes.
These parameters $\alpha, \beta, \Omega$ are independent.
Each agent $v \in V$ is assigned a strategy $s_v$ and type $\tau_v$ at $t=0$ according to parameters $\alpha$, $\beta$, and $\Omega$.
In subsequent iterations, each agent $v$ uses this strategy $s_v$ as well as the current state of the network $G_t$ to make changes to their neighborhood.
The type of each agent $\tau_v$ is public, but the strategy $s_v$ is private.
For tractability we assume that there are only two categorical values of $\tau$.

\renewcommand{\arraystretch}{1.4}

\begin{table}[]
    \caption{Utility functions $U_v^a$ and $U_v^s$ for each strategy
    }
    \label{tab:util_funcs}
\begin{tabularx}{\linewidth}{@{}c c X @{}}
\textsc{Strategy} & \textsc{Utility Function} & \textsc{Description} \\ \midrule
         $H_m$ & $|N^\tau(v)|/\kappa$ & Neighbors with attribute $\tau_v$ \\
         $H_r$ & $|N(v) \setminus N^\tau(v)|/\kappa$ & Neighbors without attribute $\tau_v$\\
         $L_e$ & $\Delta_v / {\kappa \choose 2}$ & Triangles where $v$ is a vertex \\
         $L_c$ & $I_v / \kappa$ & Independent neighbors of $v$ \\
\end{tabularx}
\end{table}

\renewcommand{\arraystretch}{1}

\subsection{Strategy Space} \label{strat_space}

We denote the set of attribute strategies $H$ and structural strategies $L$.
Each agent has strategy $s_v = (s^a_v, s^s_v)$ drawn from $H \times L$.
The homophilic strategy is denoted $H_m$ and the heterophilic strategy $H_r$.
Homophilic agent $v$ derives attribute utility from neighbors of the same type as $v$. Heterophilic agent $v$ derives attribute utility from neighbors of a different type.
Let $N^\tau(v) = \{ u | u \in N(v), \tau_u = \tau_v \}$, the set of $v$'s neighbors who have the same type as $v$.
Then $N(v) \setminus N^\tau(v)$ gives the set of neighbors of $v$ of a different type.

We denote the strategy of embedded agents $L_e$ and that of social capital agents $L_c$. 
Agent $v$ desiring embeddedness derives structural utility from a connection to $u$ if the edge between $v$ and $u$ is part of one or more triangles in the network.
An agent $v$ desiring social capital gets structural utility from neighbors who are disconnected.
Let $\Delta_v = |\{ (u,w) | u, w \in N(v), (u, w) \in E \}|$, the number of triangles where $v$ is a vertex.
Additionally let $I_v = |\{ u | u \in N(v), \not\exists w \text{ s.t. } w \in N(v), (u,w) \in E\}|$, the number of components of size 1 in the subgraph induced by $v$'s neighborhood $N(v)$ on $G$. 

We give detailed utility functions in Table \ref{tab:util_funcs}. 
Agent $v$'s strategy $s_v$ is drawn from $\{H_m, H_r\} \times \{L_e, L_c\}$.
The aggregate utility of $v$ is the sum of the attribute and structural utility, $U_v = U_v^a + U_v^s$.
We normalize attribute utilities and social capital utility by $\kappa$, the degree constraint for all nodes, and the embedded utility by $\kappa \choose 2$ to ensure they all remain in $[0, 1]$.
Thus, $U_v \in [0, 2]$, for all agents $v$. 

\subsection{Model Description} \label{desc}

Consider the network at iteration $t$, $G_t = (V, E_t)$, with agents $V$ and edges $E_t$.
At time $G_0$ we have a trivial network, $E_0 = \emptyset$.
We drop the $t$ subscript where the meaning is clear.
Each iteration consists of two stages: first the proposal stage, and then the action stage.
Agents within each stage move simultaneously.
We give the algorithm details below, followed by a simplified algorithm.
We visually depict the algorithm in Figure \ref{fig:tikz}.

\begin{figure*}[t]
      \centering
      \begin{tikzpicture}
        \Vertex[x=0, y=0, label=$v$,RGB,color={100,190,244}]{v}
        \Vertex[x=0, y=0, RGB, opacity = 0, size = 2.75, style = dashed]{nbhd1}
        \Vertex[x=0, y=0, RGB, opacity = 0, size = 1.55, style = dashed, label = $N(v)$, position= below]{nbhd2}

        \Text[x=0,y=-2.5]{Figure \ref{fig:tikz}a: Proposal stage at time $t$}

        \Vertex[x=.75, y=.75,  RGB,color={142,142,142}]{n1}
        \Vertex[x=-.75, y=.75,  RGB,color={142,142,142}]{n4}
        \Vertex[x=.75, y=-.75,  RGB,color={142,142,142}]{n3}
        \Vertex[x=-.75, y=-.75,  RGB,color={142,142,142}]{n2}
        \Vertex[x=2, y=1, RGB,color={142,142,142}]{n11}
        \Vertex[x=1, y=1.75, RGB,color={142,142,142}]{n12}
        \Vertex[x=1, y=-1.75,  RGB,color={142,142,142}]{n31}
        \Vertex[x=2, y=-1,  RGB,color={142,142,142}]{n32}
        \Vertex[x=-1, y=-1.75,  shape = diamond, RGB,color={142,142,142}]{n21}
        \Vertex[x=-2, y=-1,  RGB,color={253,170,97}, shape = diamond, size = .7, label = $w$]{w}
        \Vertex[x=-2, y=1,  RGB,color={142,142,142}]{n41}
        \Vertex[x=-1, y=1.75,  RGB,color={142,142,142}]{n42}

        \Vertex[x=3, y=0,  RGB,color={100,190,244}, label = $u$]{u}

        \Edge[lw=1](v)(n1)
        \Edge[lw=1](v)(n2)
        \Edge[lw=1](v)(n3)
        \Edge[lw=1](v)(n4)
        \Edge[lw=1](n11)(n1)
        \Edge[lw=1](n21)(n2)
        \Edge[lw=1](n31)(n3)
        \Edge[lw=1](n41)(n4)
        \Edge[lw=1](n12)(n1)
        \Edge[lw=1](w)(n2)
        \Edge[lw=1](n32)(n3)
        \Edge[lw=1](n42)(n4)

        \Edge[Direct, bend = -10](v)(u)
        \Edge[Direct,  bend = 25](w)(v)

        \Vertex[x=6.5, y=0, label=$v$,RGB,color={100,190,244}]{vc}
        \Vertex[x=6.5, y=0, RGB, opacity = 0, size = 2.75, style = dashed]{nbhd1c}
        \Vertex[x=6.5, y=0, RGB, opacity = 0, size = 1.55, style = dashed, label = $N(v)$, position= below]{nbhd2c}
        
        \Text[x=6.5,y=-2.5]{Figure \ref{fig:tikz}b: Action stage at time $t$}

        \Vertex[x=7.25, y=.75,  RGB,color={142,142,142}]{n1c}
        \Vertex[x=5.75, y=.75,  RGB,color={142,142,142}]{n4c}
        \Vertex[x=7.25, y=-.75,  RGB,color={142,142,142}]{n3c}
        \Vertex[x=5.75, y=-.75,  RGB,color={142,142,142}]{n2c}
        \Vertex[x=8.5, y=1, RGB,color={142,142,142}]{n11c}
        \Vertex[x=7.5, y=1.75, RGB,color={142,142,142}]{n12c}
        \Vertex[x=7.5, y=-1.75,  RGB,color={142,142,142}]{n31c}
        \Vertex[x=8.5, y=-1,  RGB,color={142,142,142}]{n32c}
        \Vertex[x=5.5, y=-1.75,  RGB,color={142,142,142}, shape = diamond]{n21c}
        \Vertex[x=4.5, y=-1,  RGB,color={253,170,97}, shape = diamond, size = .7, label = $w$]{wc}
        \Vertex[x=4.5, y=1,  RGB,color={142,142,142}]{n41c}
        \Vertex[x=5.5, y=1.75,  RGB,color={142,142,142}]{n42c}

        \Vertex[x=9.5, y=0,  RGB,color={100,190,244}, label = $u$]{uc}

        \Edge[lw=1](vc)(n1c)
        \Edge[lw=1](vc)(n2c)
        \Edge[lw=1](vc)(n3c)
        \Edge[lw=1](vc)(n4c)
        \Edge[lw=1](n11c)(n1c)
        \Edge[lw=1](n21c)(n2c)
        \Edge[lw=1](n31c)(n3c)
        \Edge[lw=1](n41c)(n4c)
        \Edge[lw=1](n12c)(n1c)
        \Edge[lw=1](wc)(n2c)
        \Edge[lw=1](n32c)(n3c)
        \Edge[lw=1](n42c)(n4c)

        \Edge[Direct, RGB, bend = -10, RGB,color={127,201,127}](vc)(uc)
        \Edge[Direct, RGB, bend = 25, RGB, color = {228, 110, 110}](wc)(vc)

        \Vertex[x=13, y=0, label=$v$,RGB,color={100,190,244}]{vb}
        \Vertex[x=13, y=0, RGB, opacity = 0, size = 2.75, style = dashed]{nbhd1b}
        \Vertex[x=13, y=0, RGB, opacity = 0, size = 1.55, style = dashed, label = $N(v)$, position= below]{nbhd2b}
        
        \Text[x=13,y=-2.5]{Figure \ref{fig:tikz}c: Network at time $t+1$}

        \Vertex[x=13.75, y=.75,  RGB,color={142,142,142}]{n1b}
        \Vertex[x=12.25, y=.75,  RGB,color={142,142,142}]{n4b}
        \Vertex[x=13.75, y=-.75,  RGB,color={142,142,142}]{n3b}
        \Vertex[x=12.25, y=-.75,  RGB,color={142,142,142}]{n2b}
        \Vertex[x=15, y=1, RGB,color={142,142,142}]{n11b}
        \Vertex[x=14, y=1.75, RGB,color={142,142,142}]{n12b}
        \Vertex[x=14, y=-1.75,  RGB,color={142,142,142}]{n31b}
        \Vertex[x=15, y=-1,  RGB,color={142,142,142}]{n32b}
        \Vertex[x=12, y=-1.75,  RGB,color={142,142,142}, shape = diamond]{n21b}
        \Vertex[x=11, y=-1,  RGB,color={253,170,97}, shape = diamond, size = .7, label = $w$]{wb}
        \Vertex[x=11, y=1,  RGB,color={142,142,142}]{n41b}
        \Vertex[x=12, y=1.75,  RGB,color={142,142,142}]{n42b}

        \Vertex[x=14.075, y=0,  RGB,color={100,190,244},label = $u$]{ub}

        \Edge[lw=1](vb)(n1b)
        \Edge[lw=1](vb)(n2b)
        \Edge[lw=1](vb)(n3b)
        \Edge[lw=1](vb)(n4b)
        \Edge[lw=1](n11b)(n1b)
        \Edge[lw=1](n21b)(n2b)
        \Edge[lw=1](n31b)(n3b)
        \Edge[lw=1](n41b)(n4b)
        \Edge[lw=1](n12b)(n1b)
        \Edge[lw=1](wb)(n2b)
        \Edge[lw=1](n32b)(n3b)
        \Edge[lw=1](n42b)(n4b)

        \Edge[lw=1](vb)(ub)

    \end{tikzpicture}
    \caption{Depiction of our proposal stage at time $t$, action stage at time $t$, and the network at time $t+1$ following these stages. Node shape represents agent attribute $\tau_v$; color is for attribute emphasis for relevant nodes. Each agent makes a proposal and responds with an action at time $t$; we omit most of these for clarity. At time $t$, agent $v$ proposes to $u$, the random revelation. Agent $w$ proposes to $v$. During the action stage, agent $v$ does not add $(w,v)$ while $u$ agrees to add $(u,v)$. Thus $u$ is now in $N(v)$.}\label{fig:tikz}
  \end{figure*}

\begin{algorithm}
\caption{Proposal and Action Stages}\label{alg:one}
\KwData{$G$ the current network}
\KwResult{$G'$ the network after actions}
$P$ an empty map from $V \to V$\;
\tcp{Proposal Stage}
\For{$v \in V$}{
  $C \gets$ the set of agents $v$ may propose to\;
  $C_\text{max} \gets$ the agent in $C$ that will maximize $v$'s utility\;
  $U_\text{max} \gets v$'s utility change after addition of $(v, C_\text{max})$\;
  \If{$U_\text{max}$ > 0 and $\delta(v) < \kappa$}{
    $P(v) \gets C_\text{max}$ \tcp{$v$ will propose to $C_\text{max}$}}
}
\tcp{Action Stage}
$B_\text{delete}, B_\text{add} \gets getBestActions(G, P)$ \tcp{get highest utility edge deletion and addition, assuming each proposal is accepted}
\For{$v \in V$}{
    \If{$\delta(v) < \kappa - 1$ and adding $(v,B_\text{add}(v))$ increases $v$'s utility more than deleting $(v,B_\text{delete}(v))$}{
        Add $(v,B_\text{add}(v))$ to $E$\;
    }
    \ElseIf{Deleting $(v,B_\text{delete}(v))$ increases $v$'s utility}{
        Remove $(v,B_\text{delete}(v))$ from $E$\;
    }
}
\end{algorithm}

\begin{description}[style=unboxed,leftmargin= .25cm]
    \item[1) Proposal Stage:] Let $d_G(u, v)$ be the distance between $u$ and $v$ in $G$.
    For each agent $v \in V$, an agent $u \in V \setminus \{v\}$ is revealed to $v$ uniformly at random, representing a chance encounter. 
    Then $v$ may choose to propose to one agent in $\{u\} \cup \{ w | d_G(v, w) = 2\}$. 
    We will denote the degree of agent $v$ as $\delta(v)$.
    Agents may only propose if their degree is less than the constraint; or $\delta(v) < \kappa$.
    Because $v$ may propose to at most one agent, they will propose to the agent giving the maximum utility.
    \item[2) Action Stage:] The choice set for agent $v$ comprises the nodes who proposed to $v$ and current neighbors of $v$; potential actions are proposal acceptance, edge deletion, or no action.
    If an agent has degree $\delta(v) \geq \kappa - 1$ and has made a proposal, we will require they take a cost non-increasing action. 
    Otherwise they will select whichever action maximizes their utility in this iteration, including doing nothing.
    If an agent has made a proposal, they assume their proposal has been accepted in the action stage.
    Each agent selects an action, and after every agent has selected, all agents execute their choices simultaneously.
    In a given iteration each agent may take one action.
\end{description}

\subsection{Stopping Criteria} \label{stop_criteria}

We define a triad of agents $u, v, w \in V$ as stable if the subgraph induced by $u, v, w$ is of maximum utility for all agents; \ie no agent would choose to drop an edge.
We stop our network formation simulations when the number of stable triads steadies.

\begin{definition}[Stable triads]
Let $G[u, v, w]$ be the subgraph of $G$ induced by vertices $u, v, w \in V$.
Agents $u, v, w$ form a stable triad if and only if one of the following cases is true (see Figure \ref{fig:stable_triads}):
\begin{enumerate}
    \item $s_u = s_v = s_w = (L_e, H_m)$ and $\tau_u = \tau_v = \tau_w$. The fully connected subgraph induced by $u$, $v$, $w$ $G[u, v, w]$, is stable. A triple of agents with a preference for embeddedness and homophily of the same type achieve maximum utility from their connection.
    \item $s_u = s_v = s_w = (L_c, H_m)$ and $\tau_u = \tau_v = \tau_w$. Then $G[u, v, w]$ with $(u,v), (u,w) \in E, (v,w) \not \in E$ is stable. A triple of homophilic agents with a preference for social capital who are on a line is also stable as $v$ connecting to $w$ would decrease the utility for all agents.
    \item $s_u = s_v = s_w = (L_c , H_r)$. Suppose $\tau_v = \tau_w$ and $\tau_u \neq \tau_v,\tau_w$. Then $G[u, v, w]$ with $(u,v), (u,w) \in E, (v,w) \not \in E$ is stable. Similar to the prior case, a triple of heterophilic social capital agents on a line is also stable.
\end{enumerate}
\end{definition}

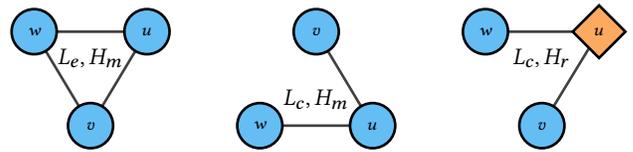
\begin{figure}[h]
      \centering
      \begin{tikzpicture}
        \Vertex[x=0, y=0, label=$v$,RGB,color={100,190,244}]{v1}
        
        \Vertex[x=.75, y=1.25,  label = $u$, RGB,color={100,190,244}]{u1}
        \Vertex[x=-.75, y=1.25,  label=$w$, RGB,color={100,190,244}]{w1}

        \Text[x=0, y=.9]{$L_e, H_m$}
        
        \Edge[lw=1](v1)(w1)
        \Edge[lw=1](v1)(u1)
        \Edge[lw=1](w1)(u1)

        \Vertex[x=3, y=1.25, label=$v$,RGB,color={100,190,244},label = $v$]{v2}
        \Vertex[x=3.75, y=0,  RGB,color={100,190,244},label = $u$]{u2}
        \Vertex[x=2.25, y=0,  RGB,color={100,190,244},label = $w$]{w2}

        \Text[x=3, y=.35]{$L_c, H_m$}
        
        \Edge[lw=1](u2)(w2)
        \Edge[lw=1](u2)(v2)

        \Vertex[x=6, y=0, label=$v$,RGB,color={100,190,244}]{v3}
        
        \Vertex[x=6.76, y=1.25,  label = $u$, shape = diamond, RGB,color={253,170,97}, size =.7]{u3}
        \Vertex[x=5.25, y=1.25,  label=$w$, RGB,color={100,190,244}]{w3}

        \Text[x=6, y=.9]{$L_c, H_r$}
        
        \Edge[lw=1](v3)(u3)
        \Edge[lw=1](w3)(u3)

    \end{tikzpicture}
    \caption{Three cases of stable triads--- agents $u,v,w$.}\label{fig:stable_triads}
  \end{figure}

\begin{theorem}
The number of stable triads in $G$ is non-decreasing.
\end{theorem}
\begin{proof}

By definition, the subgraph induced by a stable triad contains nodes at maximum utility. Thus, alteration of the triad by edge addition or deletion would strictly harm all agents.
Note that for $s_u, s_v \in \{H_m\} \times L$ and $\tau_u = \tau_v$, if $(u,v) \in E$ then removing $(u,v)$ decreases $|N^\tau(u)|$ by $1$ for $u$ and decreases $|N^\tau(v)|$ by 1 for $v$.

For the following cases we consider any stable triad $u, v, w$. When we consider alteration of one edge (\eg (u,v)), the proof extends to any equivalent edge (\eg (u,w)).

\textit{Case 1}: $s_u = s_v = s_w = (L_e, H_m)$.
Let us consider the deletion of $(u,v)$.
As $\tau_u = \tau_v$ this gives a reduction of $\frac{1}{\kappa}$ in $U_u^a$ and $U_v^a$.
Furthermore this must decrease $\Delta_u$, $\Delta_v$, and $\Delta_w$. 

\textit{Case 2}: $s_u = s_v = s_w = (L_c , H_m)$ or $s_u = s_v = s_w = (L_c , H_r)$. The heterophily case is symmetric.
Consider the deletion of $(u,v)$.
As $\tau_u = \tau_v$ this gives a $\frac{1}{\kappa}$ reduction in $U_u^a$ and $U_v^a$.

Suppose there exists an agent $x$ such that $(u,x), (v,x) \in E$.
Deleting $(u,v)$ increases the structural utility of $u$ by at most $\frac{1}{\kappa}$ as $x$ may now be an isolated neighbor of $u$.
Because agents may make at most two connections in an iteration, 
if such a $x$ were to exist then either $(u,x)$ was added first, $(v,x)$ was added first or $(u,x)$ and $(v,x)$ were added simultaneously.
If WLOG $(u,x)$ was added first then $(v,x)$ would provide $v$ with at most $\frac{1}{\kappa}$ attribute utility and reduce structural utility by the same.
Then it must be that $(u,x)$ and $(v,x)$ are added simultaneously.
However, because agents assume their proposals are accepted when evaluating candidates, there is no expected gain in $U_u$ or $U_v$.
Thus no such $x$ may exist.

Next we consider the addition of $(v,w)$.
Because $v$ and $w$ are both connected to $u$, $I_v$ and $I_w$ would both be reduced by the addition of $(v,w)$, reducing $U_v^s$, $U_w^s$, and $U_u^s$.
\end{proof}

As the number of stable triads in $G$ is non-decreasing and $G$ is finite, the number of stable triads must eventually converge.
In simulation we consider the stable triad count to stabilize when standard deviation of the prior $\kappa$ iterations is less than some small $\varepsilon$. We choose $\kappa$-many iterations because this would allow for each agent to completely change their neighborhood.

Though we do not prove this, we conjecture that convergence occurs in $O(\text{poly}(|V|, \kappa))$.
Consider a network with countably infinite vertices.
If $\kappa = 1$, it could not be the case that the network does not converge; in fact it would converge very quickly.
Even for larger values of $\kappa$, maximum degree agents in stable triads are unlikely to continually switch between neighbors, as stable triad neighbors grant high payoff.
By results from the coupon collector's problem, because of uniformly random revelations, agents will quickly find neighbors who grant them attribute utility.
In practice our simulations do not require a large number of iterations.

\graphicspath{{./figures}}

\section{Simulations}
In this section we describe our first contribution: formation of evolutionarily stable communities. We explore the parameter space via simulation to demonstrate our model's ability to create networks with varied stable community structures.

In order to effectively explore the parameter space, we fix some parameters, while varying others.
We use classical sociological results to fix values for $|V|$ and $\kappa$. From Dunbar \cite{Dunbar2009}, people can maintain active communities of $\approx150$; we fix $|V|=150$ for our simulations. Dunbar also claims that a person's 5 strongest connections are their ``support clique'', or core group of emotional support. The closest $15$ are the ``sympathy group'', who provide high cost support and time.
We are interested in active relationships, though not necessarily restricted to only the closest connections. Thus we fix $\kappa = 10$.
We also fix $\Omega$ as a uniform distribution over two attributes for simplicity.
We independently vary $\alpha$ and $\beta$ over a $9\times9$ parameter space grid: $[0,1]\times[0,1]$.
These parameters most directly control community structure and assortativity.
To account for variation due to chance encounters in the proposal stage, we run each grid cell ten times and report the average metrics.

\captionsetup[figure]{font=small,skip=0pt}
\begin{figure}[b]
    \begin{subfigure}[t]{.235\textwidth}
        \centering
    \includegraphics[width = \textwidth]{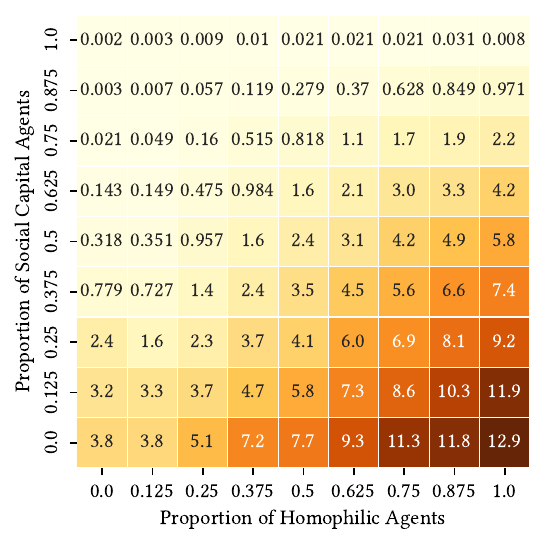}
    \caption{Triangle Count (hundreds)}
    \label{fig:std_tri}
    \end{subfigure}
    \begin{subfigure}[t]{.235\textwidth}
        \centering
    \includegraphics[width = \textwidth]{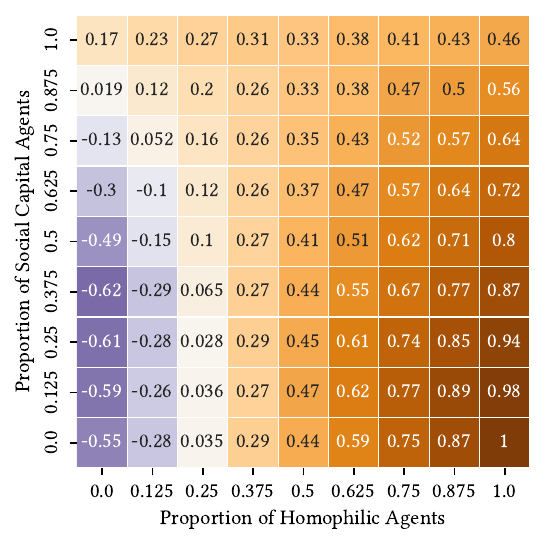}
    \caption{Assortativity Coefficient}
    \label{fig:std_assort}
    \end{subfigure}
    \begin{subfigure}[t]{0.235\textwidth}
         \centering
         \includegraphics[width=\textwidth]{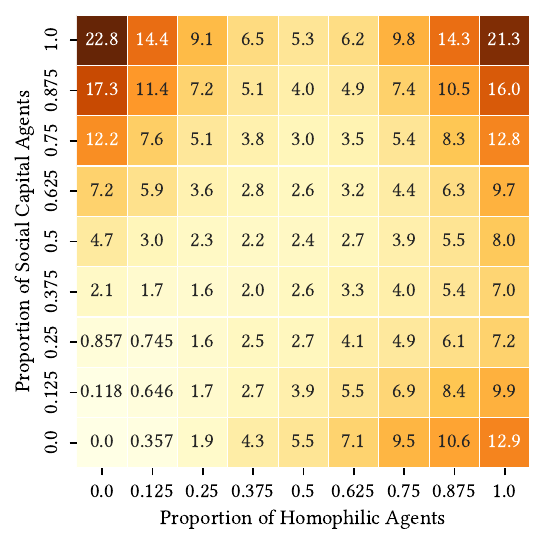}
         \caption{Stable Triad Count (hundreds)}
         \label{fig:std_stable_triad}
     \end{subfigure}
     \begin{subfigure}[t]{0.235\textwidth}
        \centering
        \includegraphics[width=\textwidth]{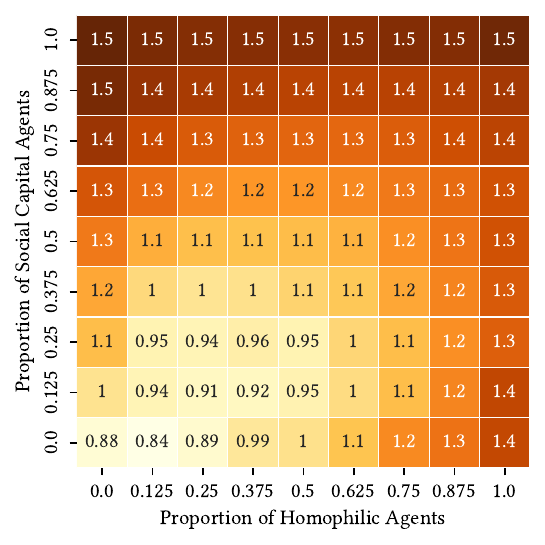}
        \caption{Average Utility of Agents}
        \label{fig:std_util}
    \end{subfigure}
    \begin{subfigure}[t]{.235\textwidth}
         \centering
         \includegraphics[width=\textwidth]{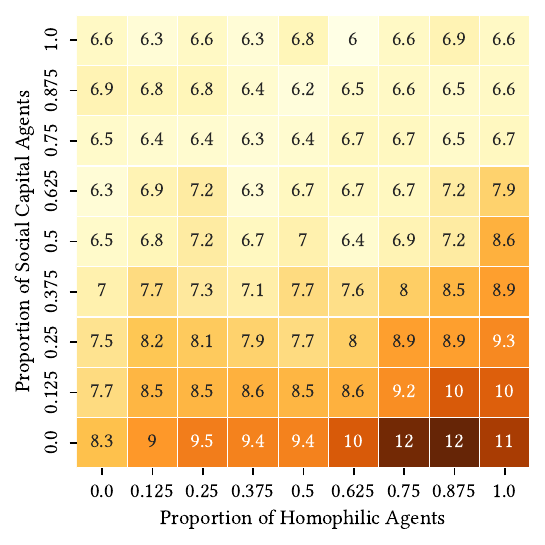}
         \caption{Number of Louvain communities}
         \label{fig:std_num_comm}
     \end{subfigure}
     \begin{subfigure}[t]{.235\textwidth}
         \centering
         \includegraphics[width=\textwidth]{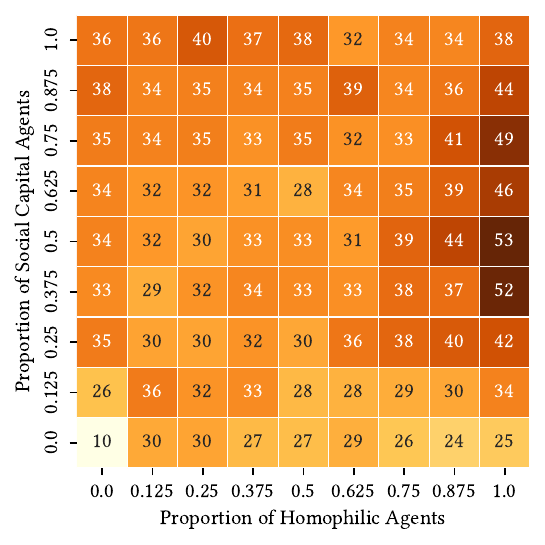}
         \caption{Terminating iteration}
         \label{fig:std_exit}
     \end{subfigure}

    \caption{Metrics over parameter space. Subfigure (b) shows that the proportion of social capital agents modulates the rate of change of the assortativity coefficient. Subfigures (c) and (d) show a bowl shape; stable communities form at the corners of our parameter space but not when embeddedness and heterophily are both high. The number of communities in (e) correlates strongly with the degree of embeddedness as agents desiring embeddedness form more neighborhoods that are closely knit.}
    \label{fig:eval_metrics}
    
\end{figure}

Triangle count depends strongly on chosen parameters.
As expected, when a higher proportion of agents follow the embeddedness and  homophily strategies, more triangles are produced (Figure \ref{fig:std_tri}).
Embedded agents' utility functions cause them to directly maximize triangle count.
The correlation with homophilous preferences is due to our use of a binary attribute --- when two agents with heterophilous preferences connect on the basis of type, a third agent may give attribute utility to at most one of these agents.

Assortativity coefficient (Figure \ref{fig:std_assort}) depends primarily on the level of homophily, which follows from the explicit optimization for assortativity and disassortativity by homophilous and heterophilous agents respectively.
There is a weaker inverse correlation in the magnitude of the assortativity coefficient based on $\beta$.
This may be due to agents preferring social capital deriving utility from any disconnected agent regardless of assortativity.

Due to the high variance in topology which results from different parameter values, our model can produce a wide variety of community structures. When agents prefer embeddedness the network forms cliques which are loosely connected. This roughly matches the idea of the strength of weak ties \cite{Granovetter1973} wherein agents may form strong, tightly knit communities which are more loosely connected. Similarly, higher preference for homophily may create more tightly knit communities as cliques of uniform type may more easily form. This is reflected by the number of communities (Figure \ref{fig:std_num_comm}) detected by the Louvain algorithm increasing with preference for embeddedness and homophily.

\section{Ablations}\label{sec:ablations}

In this section we present results for our second contribution: identification of necessary agent behaviors. 
Our model relies upon boundedly-rational, resource constrained agents pursuing attribute and structural objectives to produce complex communities.
Thus we utilize ablation tests to determine the necessity of each of these properties.
Each ablation test removes only one property.

\begin{description}[style=unboxed,leftmargin=.35cm]
\item[No resource constraint:] Removal of resource constraint. Effectively, $\kappa = |V| - 1$; an agent can form as many edges as possible.
\item[Global knowledge:] Removal of local knowledge constraint. Agents are able to propose to anyone in the network, not just those within distance two and those who are revealed to them.
\item[Ignore attribute:] Removal of payoff from homophily and heterophily. Agents solely care about structural properties.
\item[Ignore structure:] Removal of payoff from social capital and embeddedness. Agents solely care about neighbors' attributes.
\end{description}

\begin{figure}[]
\captionsetup{justification=centering} 
    \begin{subfigure}{.48\columnwidth}
    \centering
        \includegraphics[width = \textwidth, trim={0cm 1cm 0cm 1cm}, clip]{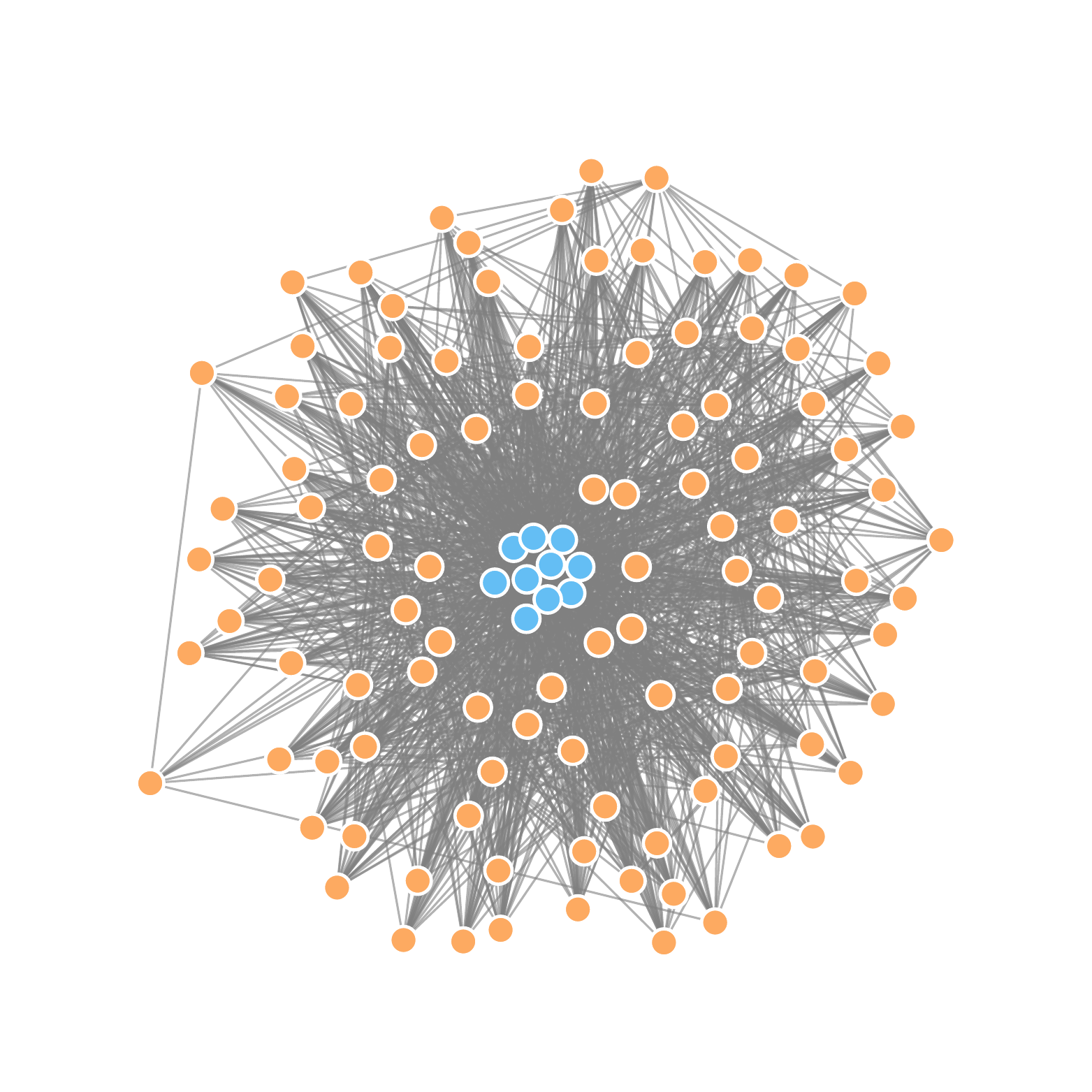}
        \caption{No Resource Constraint}
        \label{fig:abl_nobudget}
    \end{subfigure}
\hfill
    \begin{subfigure}{.48\columnwidth}
    \centering
        \includegraphics[width = \textwidth, trim={0cm 1cm 0cm 1cm}, clip]{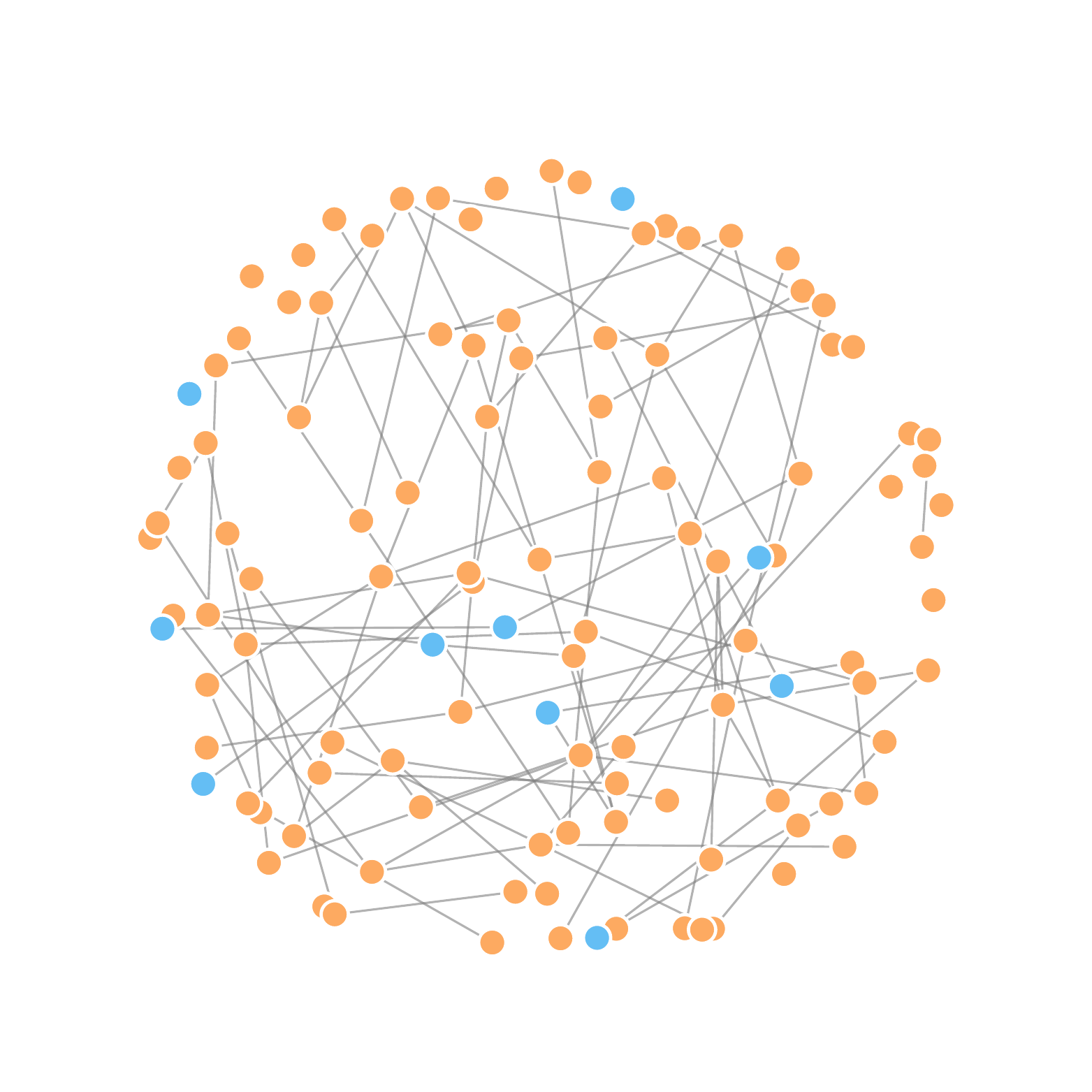}
        \caption{Global Knowledge}
        \label{fig:abl_nolocal}
    \end{subfigure}
 \begin{subfigure}{.48\columnwidth}
    \centering
        \includegraphics[width = \textwidth, trim={0cm 1cm 0cm 1cm}, clip]{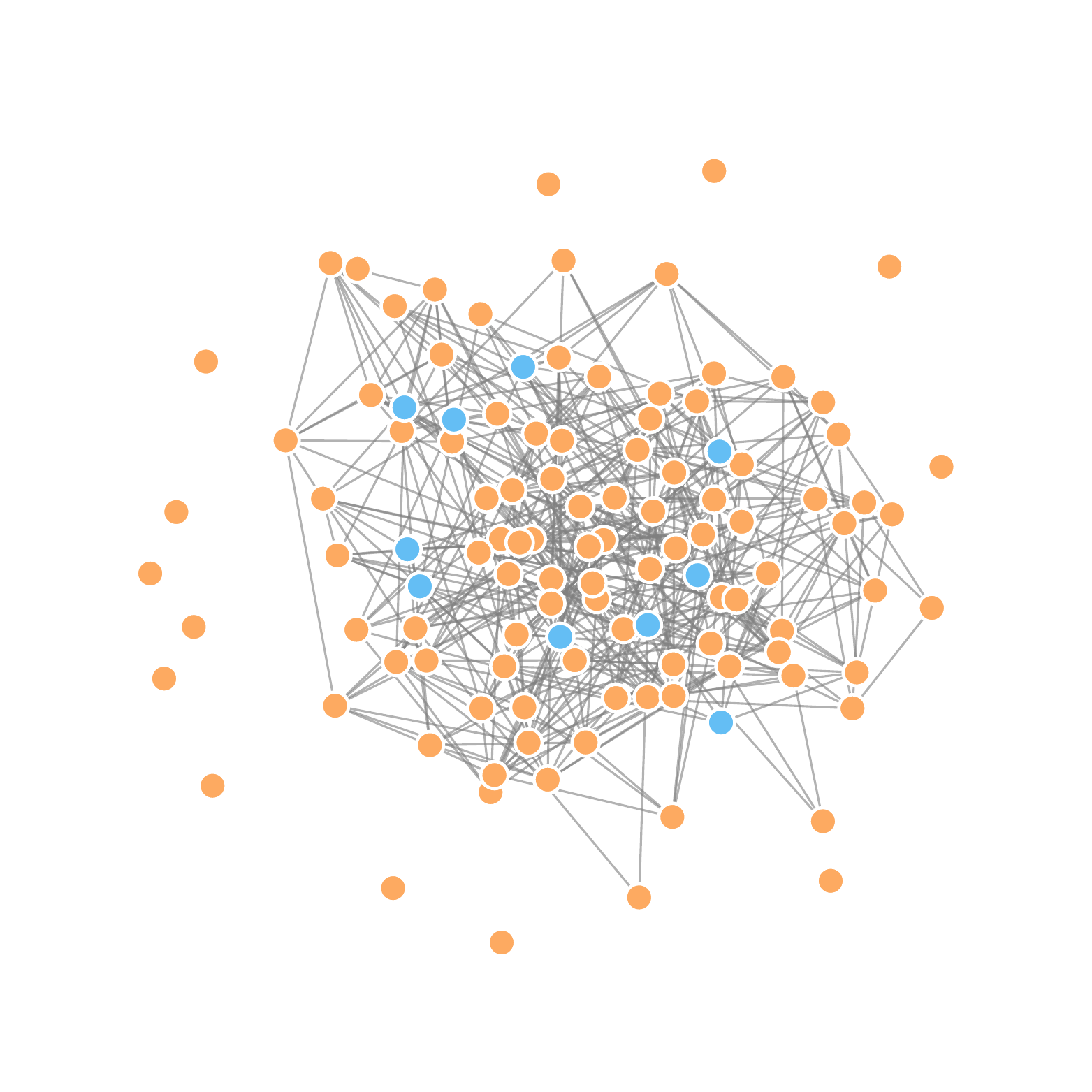}
        \caption{Ignore Attribute}
        \label{fig:abl_noattr}
    \end{subfigure}
\hfill
    \begin{subfigure}{.48\columnwidth}
    \centering
        \includegraphics[width = \textwidth, trim={0cm 1cm 0cm 1cm}, clip]{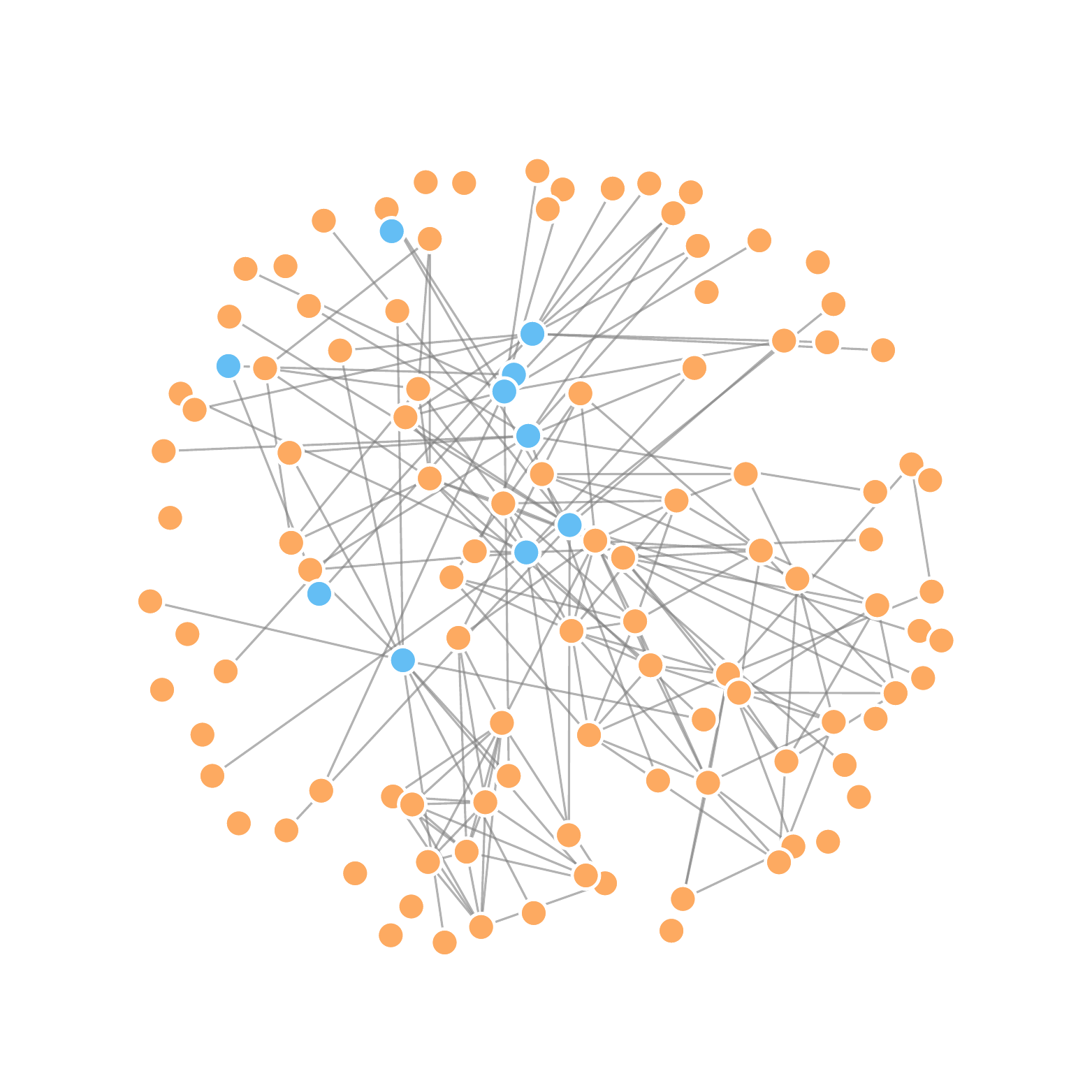}
        \caption{Ignore Structure}
        \label{fig:abl_nostruct}
    \end{subfigure}

    \caption{Ablation networks fit to Village 6. Color represents attribute. Notice the dense network in (a) due to lack of budget. Refer to Figure \ref{fig:comparison_networks}(a) and \ref{fig:comparison_networks}(b) for the actual network and our fit.}
    \label{fig:ablation_networks}
    \vspace{-1em}
\end{figure}

\begin{table}
 \caption{Ablation fits for two small villages with good standard model fit (top) and bad standard model fit (bottom).}
  \centering
  \subfloat{
    \begin{tabularx}{\linewidth}{@{}X c c c@{}}
         \textsc{Village 6 models} & \textsc{Best} $\alpha$ & \textsc{Best} $\beta$ & \textsc{\textbf{Loss}} \\
         \midrule
        No resource constraint & 0.2500 & 1.0000 & 0.1770 \\
        Global knowledge & 0.3750 & 0.5000 & 0.0843 \\
        Ignore attribute & 0.3750 & 0.5000 & 0.0658 \\
        Ignore structure & 0.3750 & 0.5000 & 0.2689 \\
        Standard (no ablation) & 0.1875 & 0.3750 & \textbf{0.0419}
    \end{tabularx}
  }
  \vspace{0.5em}
  \subfloat{
    \begin{tabularx}{\linewidth}{@{}X c c c@{}}
         \textsc{Village 10 models} & \textsc{Best} $\alpha$ & \textsc{Best} $\beta$ & \textsc{\textbf{Loss}} \\
         \midrule
        No resource constraint & 0.5000 & 1.0000 & 0.2643 \\
        Global knowledge & 0.5000 & 0.5000 & 0.1863 \\
        Ignore attribute & 0.2500 & 0.1250 & \textbf{0.0866} \\
        Ignore structure & 0.3750 & 0.5000 & 0.2552 \\
        Standard (no ablation) & 0.4375 & 0.2500 & 0.1228
    \end{tabularx}
  }
  \label{tab:ablation_fits}%
  \vspace{-1em}
\end{table}

Our ablations reveal the strength of our ecologically valid model (see Table \ref{tab:ablation_fits}, ablation networks in Figure \ref{fig:ablation_networks}; original network and our model fit in Figure \ref{tab:baseline_fit} ). The standard model generally outperforms the ablation models in both Village 6 (where our standard model achieves low loss) and Village 10 (where our standard model loss is higher). The standard model is outperformed only by the ablation where attribute is ignored for Village 10.

Here, the lower loss comes from the ablation model more accurately matching the distribution of assortativity coefficients.
The better performance of the ablation is possibly due to the model's ability to produce agent neighborhoods with mixed attributes. Our standard model assigns each agent a pure strategy of homophily or heterophily, which makes achieving a mixed neighborhood difficult.
Though agents in the global knowledge ablation have more information, it doesn't capture the fact that individuals in real networks are not omniscient, and likely make friends through mutual connections. Removing resource constraints produces a more dense network than real networks since individuals in real networks do not have the resources to befriend everyone.  Similarly, ignoring structure does not produce sufficiently dense neighborhoods.

Further modeling subtleties may be introduced to increase our accuracy.
A strict preference per agent for structure or attribute is likely too coarse; agents may exhibit different preferences over different subsets of the population.
A vector of attributes would more accurately capture true attribute spaces with utility derived from vector similarity or dissimilarity.
Agents could have homophilous or heterophilous preferences over different subsets of the attributes.

\section{Experiments}

In this section we discuss our third contribution: interpretability. We evaluate our model on real network datasets and compare our results to baselines. Unlike the baseline models, we are able to identify population-level mixing proportions over strategies which help us explain and predict the network outcome.

\subsection{Dataset}

We carefully chose a network dataset that best captured the intended use of our model. We require that a network be: attributed, undirected, and social, with active, voluntary edges.
The village microfinance dataset collected by Banerjee et al. \cite{banerjee2013diffusion} meets our inclusion criteria. We use this dataset to demonstrate that our model can be fit to 75 unique social attributed networks.

The village network data comes from a survey of social networks in 75 Indian villages.
The village networks are disjoint and the dataset contains demographic information for each household.
We choose to observe the network at the level of the household rather than the individual as survey eligibility was determined by the presence of a head of household finances.
Additionally this contraction avoids spurious homophily, heterophily, and triangles.

Banerjee et al. \cite{banerjee2013village} provide twelve edge-sets for each of the 75 villages, corresponding to household survey questions
(\eg who would you borrow money from, give advice to, etc.). 
Though the questions are directed, the authors provide undirected edges. 
We desire edges that encompass a wide range of social interactions, as our model is a general social network formation model. To ensure this, we consider the union over all edge-sets \textit{minus} any edges that involve exchange of money or goods. We exclude monetary edges because these are not purely social interactions; they require some resources outside what is required for friendship.

We use the survey data to assign household attributes. We assume that number of rooms in a home is a proxy for wealth. We normalize this by dividing the number of rooms by the square root of number of beds + 1 \cite{Hussain2008}. We desire a binary attribute for wealth for each household, so we categorize households as having $\leq 2$ normalized beds, or $>2$. We choose this cutoff because it splits the households most evenly into two categories over all villages.

\subsection{Baseline Models}

Here we describe the baseline models; we attempted a best faith implementation of both. While Mele \cite{mele2022structural} provided code, neither it nor the required packages were documented enough to utilize.

\begin{description}[style=unboxed,leftmargin=.35cm]

\item[Mele \cite{mele2022structural}]
This community formation model places attributed agents in unobserved communities at $t=0$.
In each iteration, a pair of agents is selected randomly with probability depending on network structure, community membership, and agents' attributes.
Agents keep or create an edge if the sum of their utilities for the edge is positive. Some random variation is introduced to model unobserved strategies.
The utility of an agent is dependent on the network structure, community membership, and agents' attributes.
The paper provides six model variations; results show that one performs best, so we utilize this model for our analysis. 
\item[Christakis et. al. \cite{christakis2020empirical}]
This community formation model has agents that use a utility function parameterized by: a preference over attributes, for homophily, for the number of common connections, for distance in the network, and over the value of an observable characteristic.
Agents have identical preferences.
One potential edge is introduced in each iteration; if the utility for the edge is positive for both agents then the edge is formed.
\end{description}

Both baselines, like our model, consider attribute and structure in formation of edges. In contrast to our own model, both baselines lack resource constraints and have hard to explain parameters. 

\subsection{Evaluation Metrics}\label{eval_metrics}
We refer to an observed network as $G^d$, and use $\widehat{G}^d$ when referring to a simulated network which attempts to replicate $G^d$.
Here we describe the metric used to evaluate this replication.
For each vertex we calculate attribute assortativity $a_i$ and local clustering coefficient $c_i$.
The assortativity coefficient $a_i$ indicates the level of homophily or heterophily exhibited by a vertex $v$, and is given by $a_i = \frac{1}{2}\left( \frac{2 N^\tau(v) - \delta(v)}{\delta(v)} + 1 \right) $.
An assortativity coefficient of one indicates complete homophily and zero complete heterophily.
The clustering coefficient of a vertex $v$ is given by $\frac{\Delta_v}{\delta(v) (\delta(v) - 1)}$.
For both networks $\widehat{G}^d$ and $G^d$, we calculate a distribution of these pairs $(a_i, c_i)$ and a weight $p_i$ which corresponds to the relative frequency of $(a_i, c_i)$ in each network.
We utilize the 1-Wasserstein distance as a measure of similarity between distributions.
The distance between points in the distribution is given by the L-2 (Euclidean) norm.

We also include a global component in our loss metric to ensure good fit on global assortativity and edge density.
A network with many edges could be similar in distribution to one with few edges if they have similar local clustering coefficients. The global component alleviates this issue.
Thus, we compute the difference between triangle count and assortativity coefficient by a symmetric mean absolute percentage error, and take the average of these errors. This makes up the global component.
Note that very sparse networks have undefined clustering coefficients which we set to zero.
We report the average of distributional and global components as the loss between the simulated $\widehat{G}^d$ and the observed network $G^d$.

\subsection{Parameter Fitting} \label{param_fit}

Given an attributed network dataset, some parameters are defined by the dataset. Our model has five free parameters: $|V|$, $\kappa$, $\Omega$, $\alpha$, and $\beta$. We look only at static datasets, so $|V|$ and $\Omega$ are fixed. The remaining three parameters need to be discovered.

We use grid search to discover the best fitting $\kappa$, $\alpha$, and $\beta$ for $G^d$. 
We consider 75 village networks, and we desire a shared $\kappa$ over all villages, since $\kappa$ represents a cognitive constraint on the ability of individuals to form edges. 
To achieve this, we run a coarse grid search to determine the best-fitting $\kappa \in \{5,10,15\}$. 
Then, using this $\kappa$, we find the best fitting $\alpha$ and $\beta$ on a finer grid. 
The parameters in the baseline models \cite{mele2022structural, christakis2020empirical} are neither explanatory nor bounded. To fit them, we execute a random search over a plausible parameter space informed by the reported parameters in the papers and then a finer search over the discovered best-fitting space.
We utilize our evaluation metrics (Section \ref{eval_metrics}) to determine best fit.

\subsection{Our Experimental Results}

We now present the results of our model fitting, and compare to the baseline models. We executed all simulations five times, computed the losses from these networks, and then averaged the losses. 

We show results from fitting our model to the microfinance dataset in Figure \ref{tab:village_fit}. We only show the five best-fitting and five worst-fitting villages, along with Village 10, due to space constraints. In our initial coarse-grained fitting, we found the optimal $\kappa$ value was 10. Note that since we also ran our simulations with $\kappa=10$, Figure \ref{fig:eval_metrics} gives an idea of how metrics change with values of $\alpha$ and $\beta$.

 \begin{table}[]
    \caption{Village Microfinance Fitting Results. We include the best-fitting (top half) and worst-fitting (bottom half) villages in order of increasing loss. We also include Village 10 for the baseline comparisons.}
    \label{tab:village_fit}
    \centering
    \begin{tabularx}{.95\linewidth}{@{}r | c r c c c c@{}}
         \textsc{Village} & $|V|$ & \textsc{Assort.} & \textsc{Tri.} & \textsc{Best} $\alpha$ & \textsc{Best} $\beta$ & \textsc{\textbf{Loss}}\\
         \midrule
        4 & 239 & -0.031 & 154 & 0.4375 & 0.5625 &\textbf{0.027} \\
        3 & 292 & 0.001 & 189 & 0.3750 & 0.5000 & \textbf{0.037} \\
        8 & 94 & 0.013 & 165 & 0.3125 & 0.2500 & \textbf{0.041} \\
        6 & 114 & -0.017 & 107 & 0.1875 & 0.3750 & \textbf{0.042} \\
        2 & 195 & 0.102 & 113 & 0.3750 & 0.5625 & \textbf{0.043} \\ \addlinespace[0.5em]
        10 & 77 & 0.041 &  152 & 0.2500 & 0.4375 & \textbf{0.123} \\ \addlinespace[0.5em]
        75 & 172 & 0.123 & 661 & 0.3750 & 0.0000 & \textbf{0.177} \\
        70 & 205 & 0.048 & 860 & 0.3125 & 0.0000 & \textbf{0.182} \\
        52 & 327 & 0.040 & 1698 & 0.5625 & 0.0000 & \textbf{0.226} \\
        51 & 251 & 0.044 & 1440 & 0.5626 & 0.0000 & \textbf{0.233} \\
        69 & 180 & 0.077 & 1197 & 0.5000 & 0.0000 & \textbf{0.247} 

    \end{tabularx}
    \end{table}

Recall that $\alpha$ and $\beta$ give the proportion of homophilic and social capital agents respectively. We expect that for villages with high assortativity, $\alpha$ will be high, and for villages with a high proportion of triangles compared to network size, $\beta$ will be low. Reviewing Figures \ref{fig:std_tri} and \ref{fig:std_assort} reminds us that $\alpha$ and $\beta$ do not cause independent variation in triangle count and assortativity, so we observe some seemingly ``sub-optimal'' values of $\alpha$ and $\beta$. For example, note the five worst-fit villages. They all have been fit with a $\beta$ value of $0.000$, indicating a high triangle count. Indeed, the triangle counts for these villages are high. However, the $\alpha$ value does not vary as predictably with assortativity, leading to higher loss values. 

\subsection{Baseline Models' Experimental Results}
\vspace{-1em}
\begin{table}[h]
    \caption{Comparison to Baseline Model Fitting Results. We are unable to fit baselines to all villages due to long running times. We report mean (X) and standard deviation (Y) in the last row. (* statistics reported for 30 villages)}
    \label{tab:baseline_fit}
    \centering
    \begin{tabularx}{\linewidth}{@{}r | c c c@{}}
         \textsc{Village} & \textsc{Our Loss} & \textsc{Chris.~\cite{christakis2020empirical} Loss} & \textsc{Mele~\cite{mele2022structural} Loss}\\
         \midrule
        6 & \textbf{0.042} & 0.168 & 0.601 \\ 
        10 & \textbf{0.123} & 0.135 & 0.503 \\
        All $X (\pm Y)$ & \textbf{0.110 $\pm$ 0.042} & 0.139 $\pm$ 0.021* & --
    \end{tabularx}
    \end{table}

\vspace{-1em}

We compare our model's fit to the baselines for a small village our model fit well (Village 6), and a small village our model fit less well (Village 10). We were unable to fit all villages because the running time of the baseline models was prohibitively long. To (locally) complete a search over $200$ parameter sets to produce the networks in Figure \ref{fig:comparison_networks}, our simulation takes 45 minutes, Christakis's 40 minutes, and Mele's over 14 hours.
We show results for Village 6 in Figure \ref{fig:comparison_networks}, and full loss results in Table \ref{tab:baseline_fit}. Though the Christakis model performs fairly well (Figure \ref{fig:comparison_networks}c), the resulting network is too dense and underestimates the number of isolated nodes in Figure \ref{fig:comparison_networks}a. On the other hand, the Mele model performs poorly, clearly underestimating the density and triangle count of the true network, while properly capturing the isolated nodes. Note that the Christakis model is parameterized by seven variables, and the Mele model by nine. It is difficult to make claims regarding populations of people using these parameter sets, even if they are well fit.   

\begin{figure}
\captionsetup{justification=centering} 
    \begin{subfigure}{.46\columnwidth}
    \centering
        \includegraphics[width = \textwidth, trim={0cm 1cm 0cm 1cm}, clip]{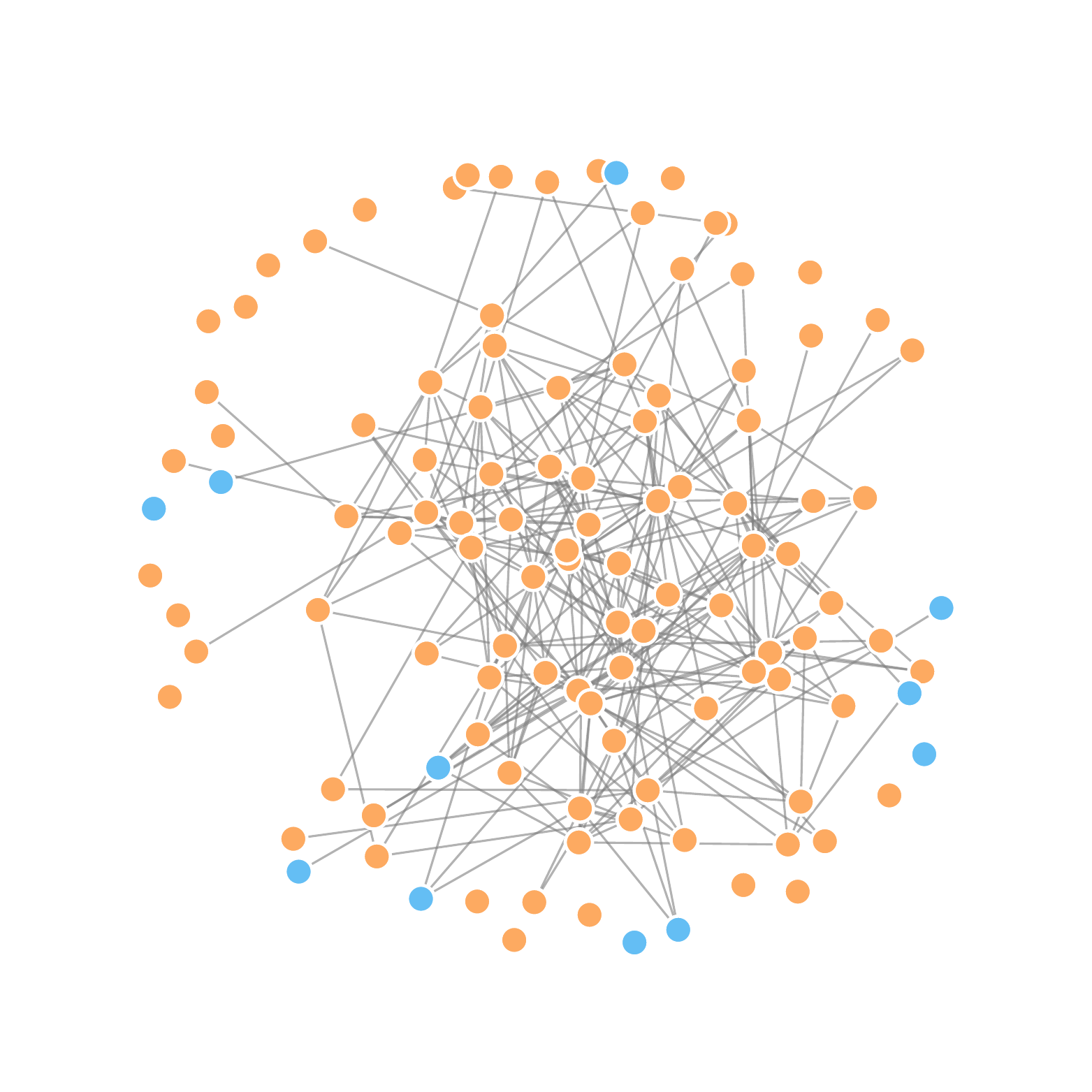}
        \caption{Actual Network}
        \label{fig:my_labela}
    \end{subfigure}
\hfill
    \begin{subfigure}{.46\columnwidth}
    \centering
        \includegraphics[width = \textwidth, trim={0cm 1cm 0cm 1cm}, clip]{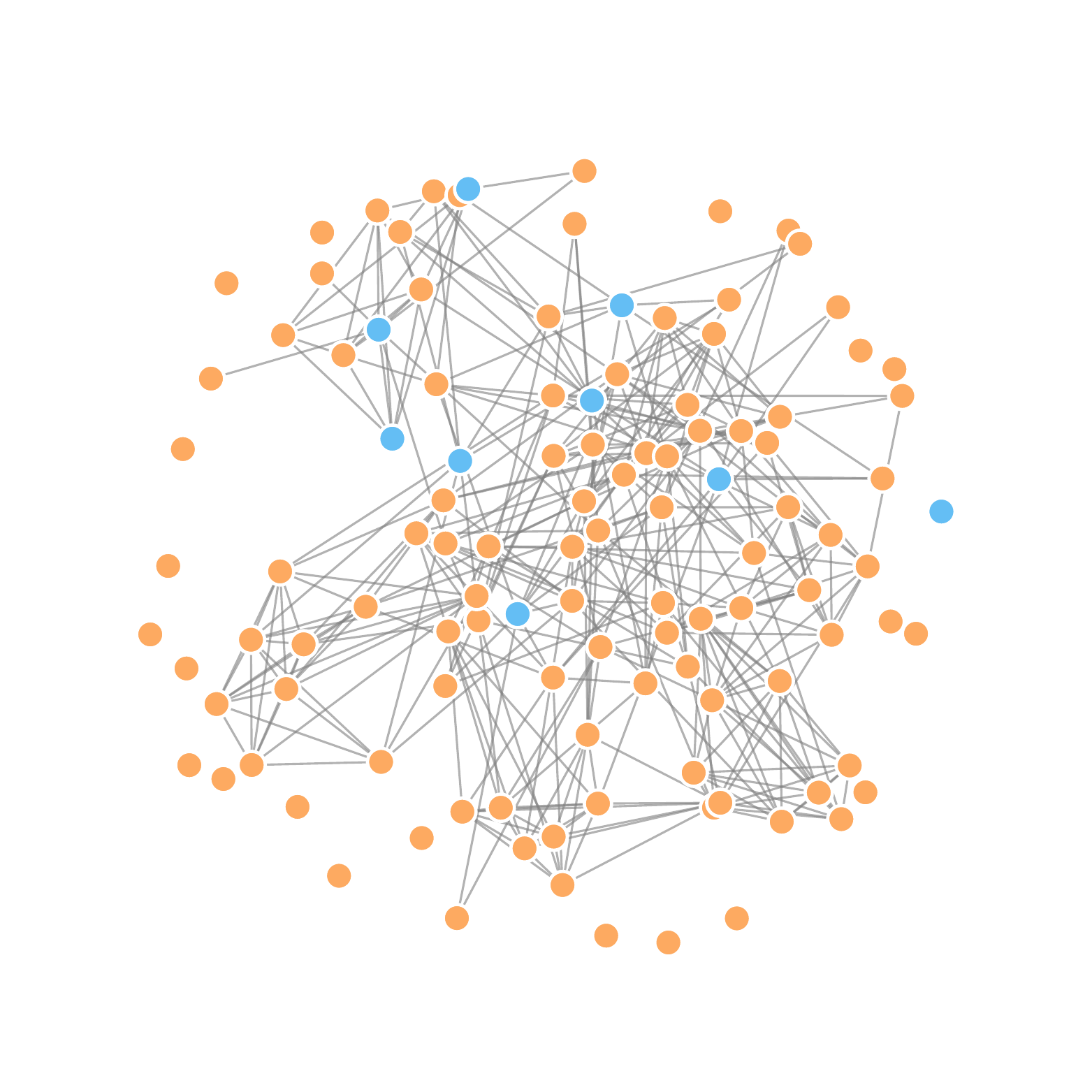}
        \caption{Our Simulation}
        \label{fig:my_labelb}
    \end{subfigure}

    \begin{subfigure}{.46\columnwidth}
    \centering
        \includegraphics[width = \textwidth, trim={0cm 1cm 0cm 1cm}, clip]{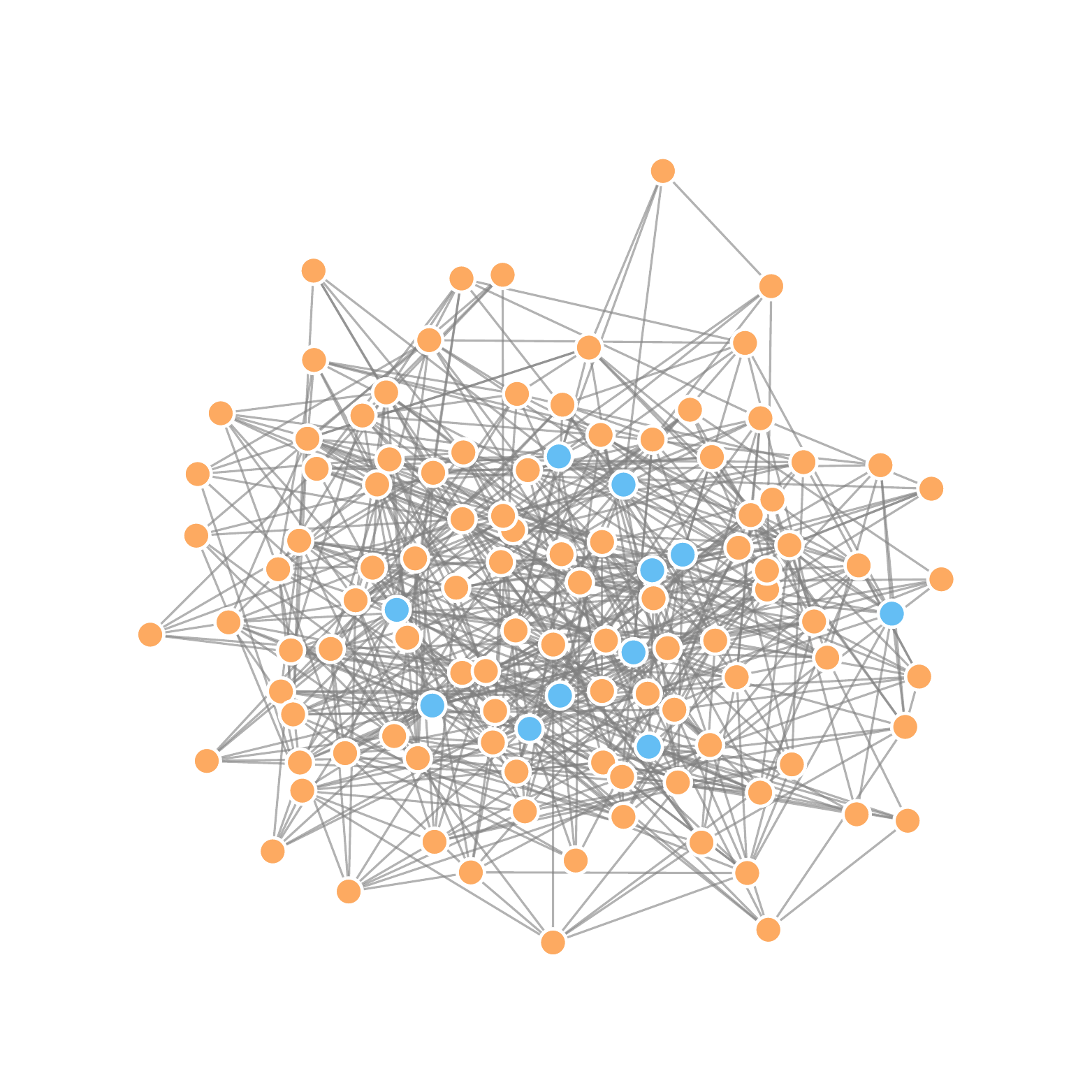}
        \caption{Christakis Simulation}
        \label{fig:my_labelc}
    \end{subfigure}
\hfill
    \begin{subfigure}{.46\columnwidth}
    \centering
        \includegraphics[width = \textwidth, trim={0cm 1cm 0cm 1cm}, clip]{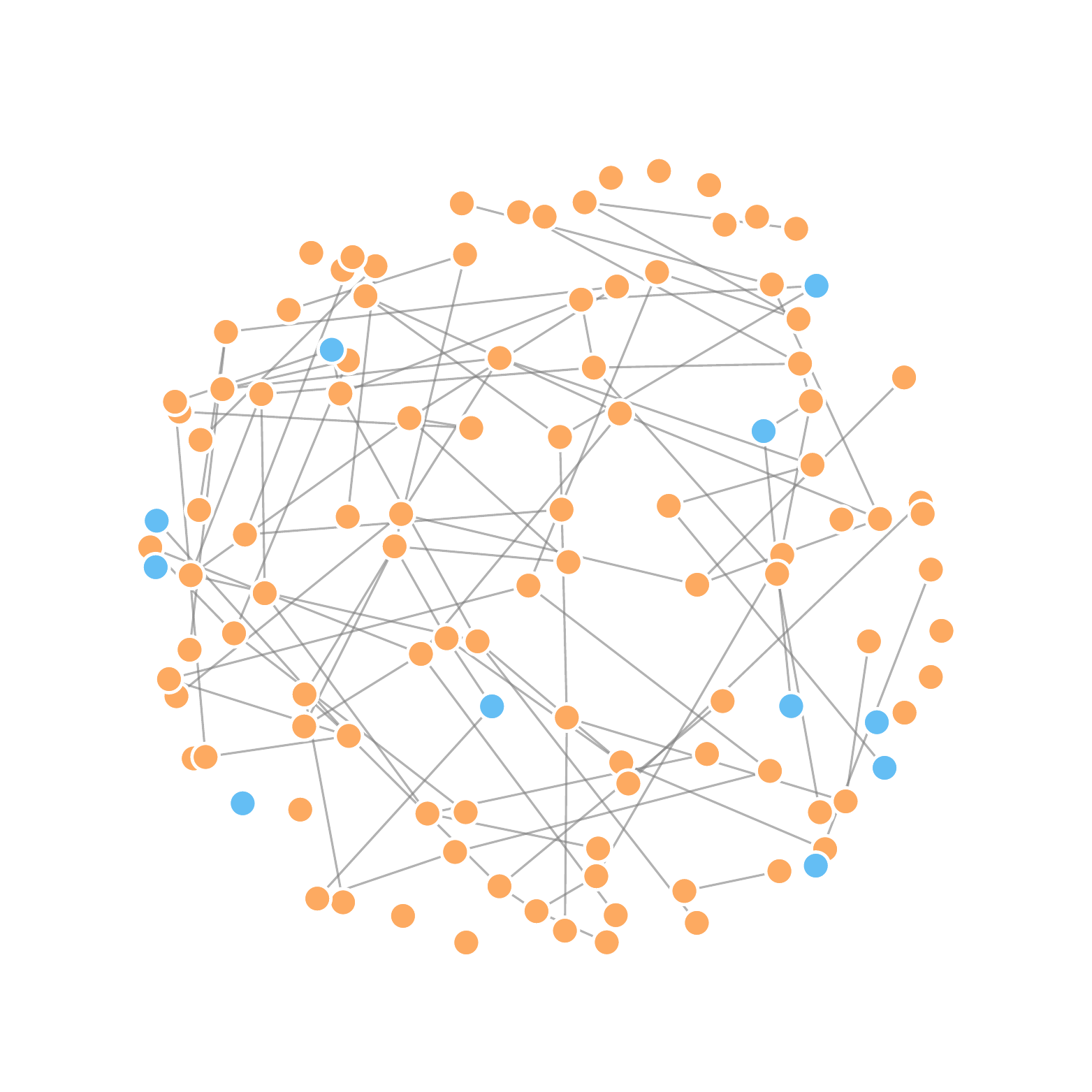}
        \caption{Mele Simulation}
        \label{fig:my_labeld}
    \end{subfigure}
    \caption{Comparison of actual network (Village 6) to our network and two baseline models. Color represents the attribute. Our model performs best, followed by Christakis~\cite{christakis2020empirical} and then Mele~\cite{mele2022structural}. }
    \label{fig:comparison_networks}
    \vspace{-1em}
\end{figure}

\section{Limitations and Future Work}
\begin{description}[style=unboxed,leftmargin=.35cm]
    \item[Scalability:]
    Agents in our model consider candidates two hops away; this can be as many as $\kappa^2$ candidates.
    Scalability of the model could be improved by restricting this candidate set.
    An implementation optimizing for sparse matrix operations could also improve computation time, especially on very large datasets.
    
    \item[Attribute dimension:]
    We only consider agents with a single scalar categorical attribute.
    This could be resolved by increasing the dimension of attribute preferences as noted in Section \ref{sec:ablations}.

    \item[Strategy space:]
    Our definition of the social capital structural utility function (requiring completely disconnected neighbors) may be overly strict, and make it difficult for social capital agents to achieve high utility.
    Instead, a modification where utility is derived per connected component in the neighborhood induced subgraph may more accurately capture real payoffs.
    
    \item[Convergence rate:]
    We stop our network formation simulations when the stable triad count stabilizes (see Section \ref{stop_criteria}).
    While we prove that this occurs in a finite number of iterations, we only conjecture a convergence rate.
    Determining when the probability of substantial structural change is low would be a stronger result.

\end{description}

\section{Conclusion}

We present an ecologically valid model of community formation with boundedly rational, attributed agents. The agents can maintain a limited number of friendships, and have local network knowledge. They employ core strategies related to local network structure and preference for homophily to form friendships. We consider a population-level mixing of strategies, including attribute and structural preferences. Every parameter value in our model is meaningful --- a particular choice of $\alpha$ (homophily), $\beta$ (social capital), $\omega$ (type distribution), and $\kappa$ (resource constraint) is a statement on the preferences of the population.
We prove that the number of stable triads in our network converges; thus, our network must become structurally stable. We show via simulation that this convergence occurs quickly in practice. Our model is well-fit to the observed networks, and ablations demonstrate ecological validity. Our model differs from the prior work in that agent behaviors are ecologically valid and interpretable.
The significance of our ecologically valid model lies in creating counterfactual scenarios and in examining policy interventions such as those designed to embed people in different parts of a social network \cite{chetty2016effects}.

\printbibliography

\end{document}